\documentclass[envcountsame]{llncs}
    \usepackage{amsmath}
    \usepackage{amsxtra}
    \usepackage{amstext}
    \usepackage{amssymb}
    \usepackage{latexsym}
    \usepackage{dsfont} 
\begin{document}
\title{A Proof of the Beierle-Kranz-Leander Conjecture related to Lightweight Multiplication in $\mathds{F}_{2^n}$}
\author{Sihem Mesnager\inst{1}  \and Kwang Ho Kim\inst{2,3} \and Dujin Jo\inst{4} \and Junyop Choe\inst{2} \and Munhyon Han\inst{2} \and Dok Nam Lee\inst{2}  } \institute{ LAGA, Department of Mathematics, University of Paris
VIII and Paris XIII,  CNRS and Telecom ParisTech, France
\email{smesnager@univ-paris8.fr} \and Institute of Mathematics,
State Academy of Sciences, Pyongyang, DPR Korea \and PGItech Corp.,
Pyongyang, DPR Korea \and  Rason Senior Middle School No.1, Rason, DPR Korea }

\maketitle

\begin{abstract}
Lightweight cryptography is a key tool for building strong security solutions for pervasive devices with limited resources. Due to the stringent cost constraints inherent in extremely large applications (ranging from RFIDs and smart cards to mobile devices), the efficient implementation of cryptographic hardware and software algorithms is of utmost importance to realize the vision of generalized computing.

In CRYPTO 2016,  Beierle, Kranz and Leander have considered lightweight multiplication in $\mathds{F}_{2^n}$. Specifically,  they have considered the fundamental question of optimizing finite field multiplications with one fixed element and  investigated which  field representation, that is which choice of basis, allows for an optimal implementation. They have left open a conjecture related to two XOR-count. Using the theory of linear algebra,  we prove in the present paper that their conjecture is correct. Consequently, this proved conjecture can be used as a reference for further developing and implementing cryptography algorithms in lightweight devices.\\

\end{abstract}

\section{Introduction}

The current  pervasive computing age has lead to an increased demand for security for applications ranging from RFIDs and smart cards to mobile devices. Lightweight cryptography is a key tool for building strong security solutions for pervasive devices with limited resources. These devices implement lightweight ciphers which are reliable and require low power and low computations. The lightweight cipher should be designed with fast encryption speed and minimal use of resources.
Due to the stringent cost constraints inherent in these extremely large applications, the efficient implementation of cryptographic hardware and software algorithms is of utmost importance to realize the vision of generalized computing. However, the computer complexity inherent in encryption algorithms poses a major challenge. Two surveys of lightweight cryptography can be found in \cite{Paar-et-al 2007} and \cite{Pawar-Pattanshetti 2018}.  In 2016, Beierle et al. \cite{BKL16} have considered Lightweight Multiplication in $\mathds{F}_{2^n}$. Specifically,  they have considered the fundamental question of optimizing finite field multiplications with one fixed element and  investigated which  field representation, that is which choice of basis, allows for an optimal implementation.

In this paper, the field $\mathds{F}_{2^n}$ is considered as the
$n$-dimensional vector space over the prime field $\mathds{F}_2$, and
given a basis $B$, every element of this vector space is uniquely
represented as a $\mathds{F}_2$-linear combination of elements in
the basis $B$. In particular, a multiplication by a fixed element
$\alpha\in \mathds{F}_{2^n}$ becomes a linear transformation over
the field $\mathds{F}_{2^n}$ and this can be identified with a
$n\times n$ matrix with entries in $\mathds{F}_2$. As this
representation of a matrix differs according to the choice of basis
of $\mathds{F}_{2^n}$, an efficiency of the multiplication with a
fixed element in $\mathds{F}_{2^n}$ depends on the choice of the
field representations, that is, the choice of $\mathds{F}_2$-basis
of $\mathds{F}_{2^n}$. The particular field representation has a
tremendous impact on the efficiency of the multiplication. Here the
efficiency of the multiplication is measured by the number of XOR
operations needed to implement the multiplication.

In \cite{BKL16}, Beierle et al. focused more specifically  on the optimal
implementation of multiplications with one given element over the
field $\mathds{F}_{2^n}$ and proposed a novel definition of
XOR-count to evaluate the efficiency of the multiplication, which is
more appropriate to consider the actual number of XOR operations
than the definition of the XOR-count proposed in \cite{FSE15}.

Considering the field representation with optimal implementation of
the multiplication by a fixed element $\alpha$ in
$\mathds{F}_{2^n}$, they present a remarkable result that the
multiplication by an element over $\mathds{F}_{2^n}$ can be
implemented with only one XOR-count if and only if the minimal
polynomial of the element is a trinomial of degree $n$. As a
generalization, they propose an open problem stated as
follows:

{\bf Conjecture} (Conjecture 1 of \cite{BKL16}): For an element
$\alpha\in \mathds{F}_{2^n}$ with two XOR-count, the minimal
polynomial $m_{\alpha}$ is of Hamming weight 5.


This open conjecture was based on their computer search to find
optimal bases for small fields of dimension smaller or equal to
eight. Their search results showed that the converse statement of
the conjecture is wrong.\\

In this paper, we prove the above conjecture. The remainder of this
paper is organized as follows. In Section 2, we introduce some
notation and definitions, and provide some propositions which are
useful for the later proofs. In Section 3 we prove that the
conjecture is correct.

\section{Preliminaries}

In this section, we introduce some notations and definitions. We
also survey some useful results in \cite{BKL16} for discussions in
next sections. Note that our description in this section closely
follows the one given in \cite{BKL16}.

\subsection{Notation and basic facts}

For a prime $p$, we denote the finite field with p elements by
$\mathds{F}_p$ and the extension field with $p^n$ elements by
$\mathds{F}_{p^n}$, respectively. In this work, we consider binary
fields, thus p = 2. Although there exists up to isomorphism only one
finite field for every possible order, we are interested in the
specific representation. For instance, if $q\in\mathds{F}_2[x]$ is
an irreducible polynomial of degree n, then $\mathds{F}_{2^n}\cong
\mathds{F}_2[x]/(q)$ where $(q)$ denotes the ideal generated by q.
The multiplicative group of a field K is denoted by $K^*$. By the
term \emph{matrix}, we refer to matrices with entries in
$\mathds{F}_{2}$. In general, the ring of $n\times n$ matrices over
a field K will be denoted by Mat$_n(K)$. The symbol $0_n$ will
denote the \emph{zero matrix} and $I_n$ will be the \emph{identity
matrix}. $E_{i,j}\in$ Mat$_n(\mathds{F}_2)$ denotes the matrix which
consists of all zeros except in the $i$-th row of the $j$-th column
for $i,j\in\{1,\ldots,n\}$. Also, $A^{(i,j)}\in
\emph{Mat}_n(\mathds{F}_2)$ denotes the $(n-1)\times (n-1)$
submatrix of $A$ formed by deleting the $i$-th row and $j$-th column
of $A$. In addition, $A^{(i_1,j_1)(i_2,j_2)}$ stands for
$\left(A^{(i_1,j_1)}\right)^{(i_2,j_2)}$ and when $i_1\neq i_2$ and
$j_1\neq j_2$, $A^{(i_1,j_1),(i_2,j_2)}$ denotes the $(n-2)\times
(n-2)$ submatrix obtained by deleting the $i_1, i_2$-th rows and
$j_1,j_2$-th columns in $A$. We denote a block diagonal matrix
consisting of $d$ matrix blocks $A_k$ as $\bigoplus_{k=1}^dA_k$. If
$P$ is a matrix with only one non-zero in each row and each column,
then $P$ is called a \emph{permutation matrix}. The characteristic
polynomial of a matrix $A$ is defined as $\chi_A:=$det$(A+\lambda
I)\in \mathds{F}_2[\lambda]$ and the minimal polynomial is denoted
by $m_A$. Recall that the minimal polynomial is the (monic)
polynomial $p$ of least degree, such that $p(A) = 0_n$. It is a
well-known fact that the minimal polynomial divides the
characteristic polynomial, thus $\chi_A(A)=0_n$. As the minimal
polynomial and the characteristic polynomial are actually properties
of the underlying linear mapping, similar matrices have the same
characteristic and the same minimal polynomial. By wt($A$), we
denote the number of non-zero entries of a matrix $A$. Analogously,
wt($q$) denotes the number of non-zero coefficients of a polynomial
$q$. For a polynomial of degree $n$
\begin{equation}\nonumber
q = x^n + q_{n-1} x^{n-1} + \cdots + q_1 x + q_0 \in
\mathds{F}_2[x],
\end{equation}
the \emph{companion matrix} of $q$ is defined as
\begin{equation}\nonumber
C_q = \left(
\begin{array}{ccccc}
0 & \        & \        & \   & \ q_0\\
1 & \ 0      & \        & \   & \ q_1\\
  & \ \ddots & \ \ddots & \   & \  \vdots\\
  & \        & \ 1      & \ 0 & \ q_{n-2}\\
  & \        & \        & \ 1 & \ q_{n-1}\\
\end{array}
\right).
\end{equation}
Now, it is well known that $\chi_{C_q} = m_{C_q} = q$.

For any two matrix $A$ and $A'$ in Mat$_n(\mathds{F}_2)$, if
$A'=TAT^{-1}$ for some invertible $T\in$ Mat$_n(\mathds{F}_2)$ then
$A$ and $A'$ are called \emph{similar} (resp.
\emph{permutation-similar} if T is a permutation matrix) and denoted
by $A\sim B$ (resp. $A\sim_{\pi} B$ for permutation-similarity).

The field $\mathds{F}_{2^n}$ can be considered as the
$n$-dimensional vector space over the field $\mathds{F}_2$, and
given a basis $B$, every element of this vector space is uniquely
represented as a $\mathds{F}_2$-linear combination of elements in
the basis $B$. In particular, a multiplication by a fixed element
$\alpha\in \mathds{F}_{2^n}$ becomes a linear transformation over
the field $\mathds{F}_{2^n}$ and this can be identified with a
$n\times n$ matrix with entries in $\mathds{F}_2$. This matrix
depends on the choice of basis of $\mathds{F}_{2^n}$, which is
denoted by $M_{\alpha, B}$. For any two bases $B$ and $B'$, there is
a invertible matrix $T$ called the \emph{matrix of basis
transformation} such that $M_{\alpha, B}=TM_{\alpha, B'}T^{-1}$.

\subsection{XOR-count and some useful propositions}

In \cite{FSE15}, it is considered that $A\in$ Mat$_n(\mathds{F}_2)$
has an XOR-count of $t$ if and only if $A$ can be written as
$A=P+\sum_{k=1}^{t} E_{i_k,j_k}$. However the such construction does
not reflect all possible matrices which can be implemented with at
most $t$ XOR operations. Authors of \cite{BKL16} provided a tight
definition for XOR-count.
\begin{definition}[\cite{BKL16}]\label{def.1}
If $t$ is the minimal number such that an invertible matrix $A$ can
be written as
\begin{equation}\nonumber
A=P\prod_{k=1}^{t}(I+E_{i_k,j_k})
\end{equation}
with $i_k \neq j_k$ for all $k$, $A$ has XOR-count of $t$, denoted
by wt$_{\oplus}(A)=t$, where $P$ is a permutation matrix.
\end{definition}

In the above definition, note that the number of factors $(I+E_{i_k,j_k})$
gives an upper bound on the actual XOR-count. In other words, if $A$
can be written as $A=P\prod_{k=1}^{t}(I+E_{i_k,j_k})$, i.e.
wt$_{\oplus}(A)=t$ then the multiplication represented by $A$ can be
implemented with at most $t$ XOR operations.

Given $\alpha \in \mathds{F}_{2^n}$, an XOR-count of $\alpha$ is
defined as the minimal XOR-count of matrices $M_{\alpha, B}$
represented the multiplication by $\alpha$ with respect to any basis
$B$.\\

Below, we introduce some propositions presented in \cite{BKL16} with
their proofs.

\begin{proposition}[\cite{BKL16}]\label{pro.1}
If $A \sim_{\pi} A'$ then $wt_{\oplus}(A) = wt_{\oplus}(A')$.
\end{proposition}
\begin{proof} Let $wt_{\oplus}(A) = t$ and $A' = QAQ^{-1}$ for a
permutation matrix $Q$ which represents the permutation $\sigma \in
S_n$. Then $A=P\prod_{k=1}^{t}(I+E_{i_k,j_k})$ and we have
$A'=QPQ^{-1}\prod_{k=1}^t(I+E_{\sigma(i_k), \sigma^{-1}(j_k)})$
since $(I+E_{i_k,j_k})Q^{-1} = Q^{-1} + E_{i_k, \sigma^{-1}(j_k)} =
Q^{-1}(I+E_{\sigma(i_k), \sigma^{-1}(j_k)})$. It follows that
$wt_{\oplus}(A') \leq wt_{\oplus}(A)$ and by reverting the above
steps we obtain $wt_{\oplus}(A) \leq wt_{\oplus}(A')$. \qed
\end{proof}
\begin{proposition}[\cite{BKL16}]\label{pro.2}
$wt_{\oplus}(A) = wt_{\oplus}(A^{-1})$.
\end{proposition}
\begin{proof} Using the fact that the matrix $I+E_{i,j}$ with $i \neq
j$ is an involution and Proposition \ref{pro.1}, we get the
result.\\
$\left(P\prod_{k=1}^t(I+E_{i_k,j_k})\right)^{-1} =
\prod_{k=t}^1(I+E_{i_k,j_k})P^{-1} \sim_{\pi}
P^{-1}\prod_{k=t}^1(I+E_{i_k,j_k})$. \qed
\end{proof}
\begin{proposition}[\cite{BKL16}]\label{pro.3}
For any $n$-dimensional permutation matrix $P$,
\begin{equation}\nonumber
P \sim_{\pi} \bigoplus _{k=1}^d C_{x^{m_k}+1}
\end{equation}
for some $m_k$ with $\sum_{k=1}^d m_k = n$ and $m_1 \geq \cdots \geq
m_d \geq 1$.
\end{proposition}
\begin{proof} It is well-known that two permutations with the same
cycle type are conjugate \cite{Abstract alg}. That is, given the
permutations $\sigma, \tau \in S_n$ as

\begin{equation}\begin{aligned}\nonumber
\sigma =
(s_1,s_2,\cdots,s_{d_1})&(s_{d_1+1},\cdots,s_{d_2})\cdots(s_{d_{m-1}+1},\cdots,s_{d_m})\\
\tau =
(t_1,t_2,\cdots,t_{d_1})&(t_{d_1+1},\cdots,t_{d_2})\cdots(t_{d_{m-1}+1},\cdots,t_{d_m})
\end{aligned}\end{equation}
in cycle notation, one can find some $\pi\in S_n$ such that
$\pi\sigma\pi^{-1}=\tau$. This $\pi$ operates as a relabeling of
indices.

Let $\sigma$ in the form above be the permutation defined by $P$.
Now, there exits a permutation $\pi$ such that $\pi\sigma\pi^{-1} =
(d_1,1,2,\ldots,d_1-1)(d_2,d_1+1,d_1+2,\ldots,d_2-1) \cdots
(d_m,d_{m-1}+1,d_{m-1}+2,\ldots,d_m-1)$. If $Q$ denotes the
permutation matrix defined by $\pi$, one obtains $QPQ^{-1}$ in the
desired form. \qed
\end{proof}

We say that any permutation matrix of this structure is in cycle
normal form. The cycle normal form of $P$ is denoted by $C(P)$. Up
to permutation-similarity, we can always assume that the permutation
matrix $P$ of a given matrix with XOR-count $t$ is in cycle normal
form, as stated in the following corollary.

\begin{corollary}[\cite{BKL16}]\label{cor.1}
\begin{equation}\nonumber
P\prod_{k=1}^{t}(I+E_{i_k,j_k}) \sim_{\pi}
C(P)\prod_{k=1}^{t}(I+E_{\sigma(i_k),\sigma^{-1}(j_k)}),
\end{equation}
for some permutation $\sigma\in S_n$. Here, we note that
$\sigma(i_k)\neq\sigma^{-1}(j_k)$ from the invertibility of the
given matrix.
\end{corollary}

The following theorem which is a main theoretical result of
\cite{BKL16} characterizes elements with a lowest XOR-count over
$\mathds{F}_{2^n}$.

\begin{theorem}[\cite{BKL16}]\label{theo.1}
Let $\alpha\in \mathds{F}_{2^n}$. Then wt$_\oplus(M_{\alpha, B}) =
1$ for some basis $B$ if and only if $m_{\alpha}$ is a trinomial of
degree $n$.
\end{theorem}

The theorem shows that one cannot hope to implement the constant
multiplication with only one XOR-count over fields
$\mathds{F}_{2^n}$ for such $n$ that there is no any irreducible
trinomial of degree $n$, for example, multiple of 8. For extension
fields of such degree, we can expect the XOR-count of 2 as optimal
implementation. In below, a conjecture suggested by authors of
\cite{BKL16} is introduced.

\begin{conjecture}[Conjecture 1 of \cite{BKL16}]\label{conj.1}
If wt$_\oplus(M_{\alpha, B}) = 2$ for some basis $B$ then
$m_{\alpha}$ is of weight smaller or equal to 5.
\end{conjecture}

Note that the converse of the conjectured statement is wrong as seen
in Tables of \cite{BKL16}.

\section{A proof of the conjecture}

In this section, we will prove that the conjecture is correct. The
following theorem, which is a major task of this work, formalizes
the conjecture again.

\begin{theorem}\label{theo.me9}
Conjecture \ref{conj.1} is true: if wt$_\oplus(M_{\alpha, B}) = 2$
for some basis $B$ then $m_{\alpha}$ is of weight smaller or equal
to 5.
\end{theorem}

As a preliminary for proving, we provide some necessary facts and
brief computations.

\begin{proposition}\label{pro.me0}
Elements with the same minimal polynomial have the same XOR-count.
\end{proposition}
\begin{proof}
Let $\alpha, \beta$ be different roots of irreducible polynomial
$f(x)\in \mathds{F}_2[x]$ of degree $m$ and let
$B=\{b_1,\dots,b_m\}$ be a $\mathds{F}_2$-basis of
$\mathds{F}_{2^m}$ such that
wt$_\oplus(\alpha)$=wt$_\oplus(M_{\alpha,B})$, i.e. which gives the
lowest XOR-count for the multiplication with $\alpha$. And let
$\sigma\in Gal(\mathds{F}_{2^m}/\mathds{F}_2)$ map $\alpha$ to
$\beta$. From the definition of the matrix of the linear
transformation,
\begin{equation}\nonumber
\left(\begin{array}{c}
\alpha b_1 \\
\vdots\\
\alpha b_m
\end{array}\right)=M_{\alpha, B}
\left(\begin{array}{c}
b_1 \\
\vdots\\
b_m
\end{array}\right)
\end{equation}
and by the action of $\sigma$, we have
\begin{equation}\nonumber
\left(\begin{array}{c}
\beta \sigma(b_1) \\
\vdots\\
\beta \sigma(b_m)
\end{array}\right)=M_{\alpha, B}
\left(\begin{array}{c}
\sigma(b_1) \\
\vdots\\
\sigma(b_m)
\end{array}\right),
\end{equation}
which shows that $M_{\alpha,B}=M_{\beta, \sigma(B)}$ since
$\sigma(B)=\{\sigma(b_1),\dots,\sigma(b_m)\}$ is also a basis of
$\mathds{F}_{2^m}$. Thus $\textrm{wt}_\oplus(\beta)\leq
\textrm{wt}_{\oplus}(\alpha)$. Similarly, we can also get
$\textrm{wt}_\oplus(\alpha)\leq \textrm{wt}_{\oplus}(\beta)$.
\qed\end{proof}

\begin{proposition}\label{pro.me10}
Let $\alpha \in \mathds{F}_{2^n}$ be algebraic element of degree
$m$, i.e.
$\mathds{F}_2(\alpha)=\mathds{F}_{2^m}\subset\mathds{F}_{2^n}$ and
let $B$ be any basis of $\mathds{F}_{2^n}$ over $\mathds{F}_2$. Then
$\chi_{M_{\alpha,B}} = (m_\alpha)^d$, where $d=n/m$.
\end{proposition}

\begin{proof}
Since $\chi_{M_{\alpha,B}}$ is independent of the choice of the
basis, it is sufficient that we consider about a specified basis. As
a such basis, we can construct as follows: let's take the polynomial
basis $\{1,\alpha,\alpha^2,\dots,\alpha^m\}$ of $\mathds{F}_{2^m}$
over $\mathds{F}_2$ and the polynomial basis
$\{1,\beta,\beta^2,\dots,\beta^d\}$ of $\mathds{F}_{2^n}$ over
$\mathds{F}_{2^m}$. Now, $B=\{1,\alpha,\dots,\alpha^m; \ \beta,
\beta\alpha,\dots,$ $\beta\alpha^m; \ \dots \ ; \
\beta^d,\beta^d\alpha,\dots, \beta^d\alpha^m\}$ becomes a basis of
$\mathds{F}_{2^n}$ over $\mathds{F}_2$ and, with respect to this
basis, we have $M_{\alpha,B}=\bigoplus_{k=1}^d C_{m_\alpha}$. Thus
$\chi_{M_{\alpha,B}} = (\chi_{C_{m_\alpha}})^d = (m_\alpha)^d$.
\qed\end{proof}

\begin{proposition}\label{pro.me11}
Let $f \in \mathds{F}_2[x]$ be a irreducible polynomial and suppose
that wt$(f) \geq 5$. Then wt$(f^d) \geq 5$ for any integer $d \geq
1$.
\end{proposition}

\begin{proof}
Let $f(x)=x^{n_1}+x^{n_2}+\cdots+x^{n_k}+1$, where $n_1
>n_2>\cdots>n_k\geq1$ and $k\geq 4$. Since wt$(f^2)=$ wt$(f)$, we can
assume that $d$ is odd. Then $f(x)^d=(x^{n_1} + x^{n_2} + \cdots +
x^{n_k}+1)^d = x^{dn_1} + dx^{(d-1)n_1+n_2} + \cdots + dx^{n_k} +
1$. Since $d$ is odd, wt$(f^d)\geq 4$. Considering that $f$ is
irreducible, $f(1)=1$ and so $f(1)^d=1$. Thus wt$(f^d)$ is a odd and
we obtain wt$(f^d)\geq 5$. \qed
\end{proof}

\begin{proposition}\label{pro.me12}
(i) The matrix $I+E_{i,j}$ with $i\neq j$ is a involution, i.e.
$(I+E_{i,j})^{-1}=I+E_{i,j}$, and \emph{det}$(I+E_{i,j})=1$.

\noindent(ii) For any matrix $A\in\,$\emph{Mat}$_n(\mathds{F}_2)$,
\emph{det}$(A+E_{i,j})=$
\emph{det}$(A)\,+\,$\emph{det}$(A^{(i,j)})$.

\noindent(iii) For any integer $i,j\in \{1,2,\ldots,n\}$,

\begin{equation}\label{eq.me1}
\emph{det}\left((C_{x^n+1}+\lambda I)^{(i,j)}\right)=\left\{
\begin{array}{ll}
\lambda^{i-j-1} & \quad (i>j)\\
\lambda^{n+i-j-1} & \quad (i\leq j)
\end{array}
\right..
\end{equation}
Furthermore, the determinant of a matrix $(C_{x^n+1}+\lambda
I)^{(i_1,j_1)(i_2,j_2)}$ is zero or $\lambda^{k}$ for some integer
$k\geq 0$.
\end{proposition}

\begin{proof}
(i) and (ii) of the proposition are trivial from the fundamental
properties of matrix theory.

\noindent (iii) We define the $m\times m$ matrices $H_m^\lambda$ and
 $S_m^\lambda$ as the following:

\begin{equation}\nonumber
H_m^\lambda := \left(
\begin{array}{ccccc}
\ \lambda &&&&\\
\ 1 & \ \lambda &&&\\
& \ \ddots & \ \ddots &&\\
&& \ 1 & \ \lambda &\\
&&& \ 1 & \ \lambda
\end{array}
\right), \quad S_m^\lambda := \left(
\begin{array}{ccccc}
1 & \ \lambda &&&\\
  & \ 1 & \ \lambda &&\\
  &   & \ \ddots & \ \ddots &\\
  &&& \ 1 & \ \lambda\\
  &&&& \ 1\\
\end{array}
\right).
\end{equation}

If $i\leq j$ then the matrix $(C_{x^n+1}+\lambda I)^{(i,j)}$ has the
following form:

\begin{equation}\nonumber
U = \left(
\begin{array}{ccr}
H_{i-1}^\lambda && 1\\
& S_{j-i}^\lambda & \\
&& H_{n-j}^\lambda\\
\end{array}
\right)
\end{equation}
and det$(C_{x^n+1}+\lambda
I)^{(i,j)}=\textrm{det}(H_{i-1}^\lambda)\cdot
\textrm{det}(S_{j-i}^\lambda)\cdot \textrm{det}(H_{n-j}^\lambda)=
\lambda^{n-j+i-1}$.

If $i > j$ then, by moving the first row to the last row in the
matrix $(C_{x^n+1}+\lambda I)^{(i,j)}$ without the change of the
determinant, we get the following matrix:

\begin{equation}\nonumber
V=\left(
\begin{array}{lcc}
S_{j-1}^\lambda &&\\
& H_{i-j-1}^\lambda & \\
\lambda && S_{n-i+1}^\lambda
\end{array}
\right).
\end{equation}
Thus, det$(C_{x^n+1}+\lambda I)^{(i,j)} = \textrm{det}\,V =
\lambda^{i-j-1}$.\\

Now we consider det$\,U^{(k,l)}$ and det$\,V^{(k,l)}$ in order to
evaluate the determinant of the matrix $(C_{x^n+1}+\lambda
I)^{(i,j)(k,l)}$.

If the $(k,l)$-entry of $U$ belongs to a diagonal block matrix
$H^{\lambda}$ or $S^{\lambda}$ of the matrix, then det$\,U^{(k,l)}$
is $\lambda^d$ for some integer $d\geq 0$. If the entry belongs to
an upper part of diagonal blocks of $U$, then det$\,U$ =
det$\,(U+E_{k,l})$ and thus
$\textrm{det}\,U^{(k,l)}=\textrm{det}\,(U+E_{k.l})-\textrm{det}\,U=0$.
In the case that the $(k,l)$-entry of $U$ belongs to a lower part of
diagonal blocks, if the $k$-th row and $l$-th column of $U$ cross to
the blocks $H_{n-j}^\lambda$ and $H_{i-1}^\lambda$ respectively,
then det $U^{(k,l)}=\lambda^d$ for some integer $d\geq 0$ because
the matrix obtained by moving the first row to the last row in
$U^{(k,l)}$ becomes a block diagonal matrix with $H^\lambda$ or
$S^\lambda$ form as diagonal blocks. Otherwise, since
\begin{equation}\nonumber
\textrm{det}\,U^{(k,l)}=\textrm{det}\left(\begin{array}{cc}
H_{i-1}^\lambda & \\
& S_{j-i}^\lambda
\end{array}\right)^{(k,l)}\cdot \textrm{det}\,H_{n-j}^\lambda
\end{equation}
or
\begin{equation}\nonumber
\textrm{det}\,U^{(k,l)}=\textrm{det}\,H_{i-1}^\lambda \cdot
\textrm{det}\left(\begin{array}{cc}
S_{j-i}^\lambda & \\
& H_{n-j}^\lambda
\end{array}\right)^{(k',l')},
\end{equation}
and
\begin{equation}\nonumber
\textrm{det}\left(\begin{array}{cc}
H_{i-1}^\lambda & \\
& S_{j-i}^\lambda
\end{array}\right)^{(k,l)} = \textrm{det}\left(\begin{array}{cc}
S_{j-i}^\lambda & \\
& H_{n-j}^\lambda
\end{array}\right)^{(k',l')}=0,
\end{equation}
we obtain $\textrm{det}\,U^{(k,l)}=0$.

In the similar way, we can also obtain the same result for
$\textrm{det}\,V^{(k,l)}$. Therefore $\textrm{det}(C_{x^n+1}+\lambda
I)^{(i,j)(k,l)}$ is zero or $\lambda^{d}$ for some integer $d\geq
0$.
\qed\end{proof}

\begin{corollary}\label{cor.me13}
The characteristic polynomial of the matrix $C_{x^n+1}(I+E_{i,j})$
is a trinomial. Precisely,
\begin{equation}\label{eq.me2}
\emph{det}\Big(C_{x^n+1}(I+E_{i,j})+\lambda I\Big)=\left\{
\begin{array}{ll}
\lambda^n+\lambda^{i-j}+1 & \quad (i>j)\\
\lambda^n+\lambda^{n+i-j}+1 & \quad (i\leq j)
\end{array}
\right..
\end{equation}
\end{corollary}

\begin{proof}
\begin{equation}\nonumber\begin{aligned}
\chi_M&=\textrm{det}(C_{x^n+1}(I+E_{i,j})+\lambda I)\\
&=\textrm{det}(C_{x^n+1}+\lambda (I+E_{i,j})) \quad \textrm{(by Proposition \ref{pro.me12}, (i))}\\
&=\textrm{det}((C_{x^n+1}+\lambda I)+\lambda E_{i,j}) \\
&=\textrm{det}(C_{x^n+1}+\lambda I)+\lambda\textrm{det}\left((C_{x^n+1}+\lambda I)^{(i,j)}\right) \quad \textrm{(by Proposition \ref{pro.me12}, (ii))}\\
&=\left\{
\begin{array}{ll}
\lambda^n+1+\lambda^{i-j} & \quad(i>j)\\
\lambda^n+1+\lambda^{n+i-j} & \quad(i\leq j)\\
\end{array}
\right. \quad \textrm{(by Proposition \ref{pro.me12}, (iii))}.
\end{aligned}\end{equation}\qed
\end{proof}

Let $\alpha\in \mathds{F}_{2^n}^{*}$ be given. As mentioned in
Section 2, by using a cycle normal form, the matrix $M_{\alpha, B}$
representing the multiplication with $\alpha$ can be written as
follows:
\begin{equation}\label{eq.me3}
M_{\alpha,
B}=\left(\bigoplus_{k=1}^{s}C_{x^{m_k}+1}\right)\prod_{k=1}^{t}(I+E_{i_k,j_k}),
\end{equation}
where $\sum_{k=1}^{s}m_k=n$ and $m_1\geq\cdots\geq m_s\geq 1$, and
$i_k\neq j_k$ for all $k$.

\begin{lemma}\label{lem.me14}
Let $\alpha\in \mathds{F}_{2^n}^{*}$ and $\alpha\neq 1$. If
\emph{wt}$_\oplus(M_{\alpha,B})=2$, then the number of blocks in the
cycle normal form is less or equal to 2, i.e. $s\leq 2$.
\end{lemma}

\begin{proof}
Since wt$_\oplus(M_{\alpha,B})=2$, we have $M_{\alpha,
B}=\left(\bigoplus_{k=1}^{s}C_{x^{m_k}+1}\right)
(I+E_{i_1,j_1})(I+E_{i_2,j_2})$. In general, the matrix
$A(I+E_{i,j})$ can be obtained by adding the $i$-th column of $A$ to
$j$-th column of $A$. Thus $M_{\alpha,B}$ can be obtained by
changing at most two columns, i.e. $j_1$-th column and $j_2$-th
column, in the matrix $\bigoplus_{k=1}^{s}C_{x^{m_k}+1}$. If we
suppose that $s\geq 3$, then neither $j_1$ nor $j_2$-th column
crosses at least one block $C_{x^{m_k}+1}$ in the matrix. Thus
$\chi_{M_{\alpha,B}}$ is divided by $\chi_{C_{x^{m_k}+1}}=x^{m_k}+1$
and so by $x+1$. Since $\chi_{M_{\alpha,B}}=(m_\alpha)^d$ for some
integer $d$ (from Proposition \ref{pro.me10}), we get $m_\alpha=x+1$
which contradicts to $\alpha\neq 1$. \qed
\end{proof}

\subsubsection{Proof of Theorem \ref{theo.me9}.}

Let $\mathds{F}_2(\alpha)=\mathds{F}_{2^m}\subset\mathds{F}_{2^n}$
and let the XOR-count of $\alpha$ be two, that is
wt$_\oplus(M_{\alpha, B})=2$ for some basis $B$ of
$\mathds{F}_{2^n}$ over $\mathds{F}_2$. From Lemma \ref{lem.me14},
$M_{\alpha,B}$ can be represented by either
\begin{equation}\label{eq.me4}
M_{\alpha,B}=C_{x^n+1}(I+E_{i_1,j_1})(I+E_{i_2,j_2})
\end{equation}
or
\begin{equation}\label{eq.me5}
M_{\alpha,B}=\left(C_{x^{m_1}+1}\bigoplus
C_{x^{m_2}+1}\right)(I+E_{i_1,j_1})(I+E_{i_2,j_2}),
\end{equation}
where $m_1+m_2=n$ and $m_1\geq m_2 \geq 1$.\\

\textbf{1) The Case that
$M_{\alpha,B}=C_{x^n+1}(I+E_{i_1,j_1})(I+E_{i_2,j_2})$}\\

In this case we first show that the Hamming weight of the
characteristic polynomial of $M_{\alpha, B}$ is less or equal to 5:
\begin{equation}\nonumber\begin{aligned}
\chi_{M_{\alpha,B}}&=\textrm{det}(C_{x^n+1}(I+E_{i_1,j_1})(I+E_{i_2,j_2})+\lambda I)\\
&=\textrm{det}(C_{x^n+1}(I+E_{i_1,j_1})+\lambda (I+E_{i_2,j_2}))\\
&=\textrm{det}(C_{x^n+1}(I+E_{i_1,j_1})+\lambda I+\lambda E_{i_2,j_2}))\\
&=\textrm{det}(C_{x^n+1}(I+E_{i_1,j_1})+\lambda I)+\lambda\, \textrm{det}\left(\big(C_{x^n+1}(I+E_{i_1,j_1})+\lambda I\big)^{(i_2,j_2)}\right)\\
&=\lambda^{n}+\lambda^{(n+)i_1-j_1}+1+\lambda\,
\textrm{det}\left(\big(C_{x^n+1}+\lambda
I+E_{\sigma(i_1),j_1}\big)^{(i_2,j_2)}\right),
\end{aligned}\end{equation}
where $\sigma$ is a permutation corresponding to $C_{x^n+1}$, which
is defined as:
\begin{equation}\nonumber
\sigma(i)=\left\{\begin{array}{ll}
i+1, & \quad i\in\{1,\dots,n-1\}\\
1, & \quad i=n
\end{array}\right..
\end{equation}

If either $\sigma(i_1) = i_2$ or $j_1=j_2$, then

\begin{equation}\nonumber\begin{aligned}
\phi(\lambda):&=\textrm{det}\left(\big(C_{x^n+1}+\lambda
I+E_{\sigma(i_1),j_1}\big)^{(i_2,j_2)}\right)\\
&= \textrm{det}\left(\big(C_{x^n+1}+\lambda
I\big)^{(i_2,j_2)}\right)=\lambda^{(n+)i_2-j_2}\\
\end{aligned}\end{equation}
and otherwise, for some indices $i'$ and $j'$ in $\{1,\dots,n-1\}$,
\begin{equation}\nonumber\begin{aligned}
\phi(\lambda)&=\textrm{det}\left(\big(C_{x^n+1}+\lambda
I\big)^{(i_2,j_2)}+E_{i',j'}\right)\\
&=\textrm{det}\left(\big(C_{x^n+1}+\lambda
I\big)^{(i_2,j_2)}\right)+\textrm{det}\left(\big(C_{x^n+1}+\lambda
I\big)^{(i_2,j_2)(i',j')}\right),
\end{aligned}\end{equation}
which is a binomial at most from Proposition \ref{pro.me12}, (iii).
Thus the Hamming weight of
$\chi_{M_{\alpha,B}}$ is smaller or equal to 5.\\

Next, we show that the extension degree of $\mathds{F}_{2^n}$ over
$\mathds{F}_{2^m}$, that is $d=n/m$, is at most two in the case
considered now.

Suppose that $d\geq 3$ in order to yield a contradiction. Then the
minimal polynomial of $\alpha$ with coefficients in $\mathds{F}_2$
\begin{equation}\label{eq.me6}
m_\alpha=x^{m}+c_{m-1}x^{m-1}+\cdots+c_1 x+1
\end{equation}
has a degree $m\leq\frac{n}{3}$ because
$\chi_{M_{\alpha,B}}=(m_\alpha)^d$.

Let $B=\{b_1,b_2,\dots,b_n\}$ be a basis which gives the form
\eqref{eq.me4}. Observing the each column of the matrix
$M_{\alpha,B}$ that represents the multiplication with $\alpha$,
non-zero in the $k (\neq j_1, j_2)$-th column is only $(k+1)$-th
entry. Thus we obtain the following a list of equalities with
respect to the basis. (We assume $j_1\leq j_2$ without loss of
generality.)

\begin{equation}\label{eq.me7}\begin{aligned}
\alpha b_1&=b_2\\
&\vdots\\
\alpha b_{j_1-1}&=b_{j_1}\\
\alpha b_{j_1+1}&=b_{j_1+2}\\
&\vdots\\
\alpha b_{j_2-1}&=b_{j_2}\\
\alpha b_{j_2+1}&=b_{j_2+2}\\
&\vdots\\
\alpha b_{n}&=b_{1}.\\
\end{aligned}\end{equation}
Now, set $\beta:=b_{j_1+1}$ and $\gamma:=b_{j_2+1}$, then from
\eqref{eq.me7}, we obtain
\begin{equation}\nonumber\begin{aligned}
(b_{j_1+1},b_{j_1+2},\dots,b_{j_2})&=(\beta, \alpha\beta,
\dots,\alpha^{j_2-j_1-1}\beta),\\
(b_{j_2+1},\dots,b_n,b_1,\dots,b_{j_1})&=(\gamma,
\alpha\gamma,\dots, \alpha^{n-(j_2-j_1)-1}\gamma).\\
\end{aligned}\end{equation}

If $j_2-j_1> \frac{n}{2}$, then by multiplying $\beta$ to both sides
of \eqref{eq.me6} and substituting $\alpha$, we obtain
\begin{equation}\nonumber
\alpha^m\beta+c_{m-1}\cdot\alpha^{m-1}\beta+\cdots+c_1\cdot\alpha\beta+\beta=0,
\end{equation}
which means that a sublist
$(\beta,\alpha\beta,\dots,\alpha^{m}\beta)$ of the list
$(b_{j_1+1},b_{j_1+2},\dots,b_{j_2})$ is linear dependent, this is
contradictory to linear independence of basis. If $j_2-j_1\leq
\frac{n}{2}$ then by multiplying $\gamma$ to both sides of
\eqref{eq.me6} we also get a contradiction that a sublist
$(\gamma,\alpha\gamma,\dots,\alpha^{m}\gamma)$ of the list
$(b_{j_2+1},\dots,b_n,b_1,\dots,b_{j_1})$ is linear dependent.\\

Therefore, $d\leq 2$ and we have $\chi_{M_{\alpha,B}}=m_\alpha$ or
$\chi_{M_{\alpha,B}}=(m_\alpha)^2$ from Proposition \ref{pro.me10}.
Thus we obtain
$\textrm{wt}(m_\alpha)=\textrm{wt}(\chi_{M_{\alpha,B}})\leq 5$.\\


\textbf{2) The case that $M_{\alpha,B}=\left(C_{x^{m_1}+1}\bigoplus
C_{x^{m_2}+1}\right)(I+E_{i_1,j_1})(I+E_{i_2,j_2})$}\\

In this case, following the proof of Lemma \ref{lem.me14}, each of
$j_1$ and $j_2$-th column of $M_{\alpha,B}$ has to cross the
different blocks in the matrix. Without loss of generality, let
$j_1$-th column cross the block $C_{x^{m_1}+1}$ and $j_2$-th one
cross the block $C_{x^{m_2}+1}$, that is, we can assume that $j_1\in
\{1,\dots, m_1\}$ and $j_2\in\{m_1+1,\dots,n\}$.

Setting $P:=C_{x^{m_1}+1}\bigoplus C_{x^{m_2}+1}$, we have
\begin{equation}\label{eq.me8}\begin{aligned}
\chi_{M_{\alpha,B}}&=\textrm{det}\,(P(I+E_{i_1,j_1})(I+E_{i_2,j_2})+\lambda
I)\\
&=\textrm{det}\,(P(I+E_{i_1,j_1})+\lambda(I+E_{i_2,j_2}))\\
&=\textrm{det}\,(P+\lambda I+E_{\sigma(i_1),j_1}+\lambda
E_{i_2,j_2}),
\end{aligned}\end{equation}
where $\sigma$ is a permutation corresponding to $P$ which is
defined as
\begin{equation}\label{eq.me9}
\sigma(i)=\left\{
\begin{array}{ll}
1 & \quad \textrm{if} \ i=m_1,\\
m_1+1 & \quad \textrm{if} \ i=n,\\
i+1 & \quad \textrm{otherwise.}
\end{array}
\right.
\end{equation}
If both $\sigma(i_1)$ and $i_2$ are in $\{1,\dots, m_1\}$ or both in
$\{m_1+1,\dots,n\}$, then from \eqref{eq.me8}
$\chi_{M_{\alpha,B}}=(m_\alpha)^d$ is divided by
$\chi_{C_{x^{m_2}+1}}=x^{m_2}+1$ or
$\chi_{C_{x^{m_1}+1}}=x^{m_1}+1$, so by $x+1$. Thus we get
$m_\alpha=x+1$, which is contradictory to $\alpha\neq 1$.\\

Now we consider two cases:\\

(i) If $\sigma(i_1)\in \{1,\dots, m_1\}$ and $i_2\in
\{m_1+1,\dots,n\}$, then we have
\begin{equation}\nonumber\begin{aligned}
&\chi_{M_{\alpha,B}}=(m_\alpha)^d=\\
&=\textrm{det}(C_{x^{m_1}+1}+\lambda I_{m_1} + E_{\sigma(i_1),j_1})
\,\cdot\, \textrm{det}(C_{x^{m_2}+1}+\lambda I_{m_2} + \lambda
E_{i_2-m_1,j_2-m_1}).
\end{aligned}\end{equation}

Thus, $\textrm{det}(C_{x^{m_1}+1}+\lambda I_{m_1} +
E_{\sigma(i_1),j_1})=(m_\alpha)^{d_1}$ for some integer $d_1$. From
Corollary \ref{cor.me13}, the left side is at most a trinomial and
by applying Proposition \ref{pro.me11}, we obtain
wt$(m_\alpha)<5$.\\

(ii) If $i_2\in \{1,\dots, m_1\}$ and $\sigma(i_1)\in
\{m_1+1,\dots,n\}$, then we have

\begin{equation}\nonumber\begin{aligned}
&\chi_{M_{\alpha,B}}=\textrm{det}(P+\lambda
I+E_{\sigma(i_1),j_1})+\lambda\,\textrm{det}\left((P+\lambda
I+E_{\sigma(i_1),j_1})^{(i_2,j_2)}\right)\\
&=\textrm{det}(C_{x^{m_1}+1}+\lambda I)\cdot
\textrm{det}(C_{x^{m_2}+1}+\lambda
I)+\\
&\hspace{2.5cm}+\lambda\,\textrm{det}\left((P+\lambda
I)^{(i_2,j_2)}\right)+\lambda\,\textrm{det}\left((P+\lambda
I)^{(i_2,j_2),(\sigma(i_1),j_1)}\right)\\
&=(\lambda^{m_1}+1)(\lambda^{m_2}+1)+\lambda\Big(\textrm{det}(P+\lambda
I+E_{i_2,j_2})-\textrm{det}(P+\lambda
I)\Big)+\\
&\hspace{6.3cm}+\lambda\,\textrm{det}\left((P+\lambda
I)^{(i_2,j_1),(\sigma(i_1),j_2)}\right)\\
&=(\lambda^{m_1}+1)(\lambda^{m_2}+1)+\\
&\hspace{1.6cm}+\lambda\,\textrm{det}\left((C_{x^{m_1}+1}+\lambda
I_{m_1})^{(i_2,j_1)}\right)
\textrm{det}\left((C_{x^{m_2}+1}+\lambda I_{m_2})^{(\sigma(i_1),j_2)}\right).\\
\end{aligned}\end{equation}
Since the last term is a monomial by Proposition \ref{pro.me12}
(iii), we obtain
\begin{equation}\label{eq.me10}
\textrm{wt}(\chi_{M_{\alpha,B}})\leq 5.
\end{equation}

On the other hand, if we suppose that $d\geq 3$ when
$\chi_{M_{\alpha,B}}=(m_\alpha)^d$, then the minimal polynomial of
$\alpha$ represented by \eqref{eq.me6} has a degree $m\leq
\frac{n}{3}$. Let $B=\{b_1,\dots,b_n\}$ be a basis generating the
matrix \eqref{eq.me5}. Then from the structure of the matrix
$M_{\alpha,B}$, we obtain a list of equalities:

\begin{equation}\nonumber\begin{aligned}
\alpha b_1&=b_2\\
&\vdots\\
\alpha b_{j_1-1}&=b_{j_1}\\
\alpha b_{j_1+1}&=b_{j_1+2}\\
&\vdots\\
\alpha b_{m_1-1}&=b_{m_1}\\
\alpha b_{m_1}&=b_{1}.\\
\end{aligned}\end{equation}

Setting $\gamma:=b_{j_1+1}$, we have
\begin{equation}\nonumber
\{b_{j_1+1},\dots,b_{m_1},b_{1},\dots,b_{j_1}\}=\{\gamma,\alpha\gamma,\alpha^2\gamma\dots,\alpha^{m_1-1}\gamma\}.
\end{equation}

By multiplying $\gamma$ to both sides of \eqref{eq.me6} and
substituting $\alpha$, we get
\begin{equation}\label{eq.me11}
\alpha^m\gamma+c_{m-1}\cdot\alpha^{m-1}\gamma+\cdots+c_1\cdot\alpha\gamma+\gamma=0,
\end{equation}
Since $m\leq n/3$ and $m_1\geq n/2$ by conditions $m_1+m_2=n$ and
$m_1\geq m_2$, the list
$(\gamma,\alpha\gamma,\dots,\alpha^{m}\gamma)$ is a sublist of the
list $(b_1,b_2,\dots,b_{m_1})$, which is linear dependent by
\eqref{eq.me11}. This is a contradiction. Thus, $d\leq 2$ and
wt$(m_\alpha)=$ wt$(\chi_{M_{\alpha,B}})\leq 5$. Consequently,
Theorem \ref{theo.me9} is proven completely. \qed

\section{Conclusions}
In CRTPTO 2016, a study of optimal multiplication bases with respect to the XOR-
count has been presented in \cite{BKL16}. The authors  have
been able to characterize exactly which field elements can be implemented with
one XOR operation only, the general case was left open. For small  fields of dimension smaller or equal to eight, they were able to compute the optimal bases with
the help of an exhaustive computer search. They have conjectured that for  an element
$\alpha\in \mathds{F}_{2^n}$ with two XOR-count, the minimal
polynomial $m_{\alpha}$ is of Hamming weight 5. In this paper we confirm the validity of this conjecture and prove it. The proved result can be used to improve the performance of the algorithms in lightweight cryptography.

\end{document}